\documentclass[journal]{IEEEtran}
\usepackage{amsmath,amsfonts}
\usepackage{algorithmic}
\usepackage{algorithm}
\usepackage{array}
\usepackage[caption=false,font=normalsize,labelfont=sf,textfont=sf]{subfig}
\usepackage{textcomp}
\usepackage{stfloats}
\usepackage{url}
\usepackage{verbatim}
\usepackage{graphicx}
\usepackage{cite}
\newfloat{algorithm}{tbp}{loa}
\providecommand{\algorithmname}{Algorithm}
\floatname{algorithm}{\protect\algorithmname}

\newcommand{\beq}{\begin{equation}}
\newcommand{\eeq}{\end{equation}}
\def\bv{\mbox{\boldmath $v$}}

\def\bu{\mbox{\boldmath $u$}}

\def\bv{\mbox{\boldmath $v$}}

\def\bx{\mbox{\boldmath $x$}}

\def\bs{\mbox{\boldmath $s$}}

\def\b1{\mbox{\boldmath $1$}}
\def\b0{\mbox{\boldmath $0$}}

\def\mA{\mbox{$\mathbf{A}$}}
\def\mD{\mbox{$\mathbf{D}$}}

\def\mX{\mbox{$\mathbf{X}$}}
\def\mD{\mbox{$\mathbf{D}$}}

\def\mI{\mbox{$\mathbf{I}$}}

\def\mL{\mbox{$\mathbf{L}$}}

\def\mP{\mbox{$\mathbf{P}$}}

\def\mS{\mbox{$\mathbf{S}$}}

\def\mU{\mbox{$\mathbf{U}$}}

\newcommand{\ds}{\displaystyle}

\floatstyle{ruled}
\newfloat{algorithm}{tbp}{loa}
\providecommand{\algorithmname}{Algorithm}
\floatname{algorithm}{\protect\algorithmname}
\newenvironment{proof}[1][Proof]{\noindent \textbf{#1.} }{\qedsymbol}
\newcommand{\qedsymbol}{\hspace{\fill}\rule{1.5ex}{1.5ex}}

\newtheorem{proposition}{Proposition}

\begin{document}
\title{Learning  Laplacian Forms for Graph Signal Processing via the Deformed  Laplacian}
%
\author{Stefania Sardellitti, \IEEEmembership{Senior Member, IEEE}
\thanks{S. Sardellitti is with the Dept. of Engineering and Sciences, Universitas Mercatorum, Rome, Italy (e-mail: stefania.sardellitti@unimercatorum.it).}
\thanks{This work was supported  by the  FIN-RIC Project 
TSP-ARK, financed by Universitas Mercatorum under grant n. 20-FIN/RIC.}
}
\maketitle
%
%
%

%
\maketitle
\begin{abstract}  
Learning the  graph Laplacian from observed data is one of the most investigated and fundamental tasks in Graph Signal Processing (GSP). Different variants of the Laplacian, such as the combinatorial, signless or signed Laplacians have been considered depending on  the type of features to be extracted from the data. The main contribution of this paper is the introduction of a parametric Laplacian, called the deformed Laplacian, defined as a quadratic matrix polynomial that provides a parametric dictionary for  graph signal processing. 
The deformed Laplacian can be interpreted as the generator of a parametric linear reaction-diffusion dynamics on graphs, capturing the interplay between diffusive coupling and  nodal reaction effects.  It is a parametric polynomial matrix that enables the design of novel topological operators  tailored to both the underlying graph structure and the observed signals. Interestingly, we show that several Laplacian variants proposed in the literature arise as special cases of the deformed Laplacian. We then develop a method to jointly learn the deformed Laplacian and  the graph signals from data, showing how its use improves signal representation across a broad class of graphs  compared to standard Laplacian forms. Through extensive numerical experiments on both synthetic  and real-world datasets, including financial and communication networks, we  assess the benefits of the proposed method in terms of graph signal reconstruction error and sparsity of the representation.

\end{abstract}

\section{Introduction}
\label{sec:intro}

A wide range of data science problems can be naturally modeled using graphs,  including social, financial, communication and biological networks \cite{newman2018networks}, \cite{posfai2016network}.   In such representations, entities of interest are associated with nodes, while edges encode pairwise relationships among them, enabling the modeling of complex interactions in structured datasets. 
Recently, Graph Signal Processing (GSP) has emerged as a powerful framework for the analysis and processing of signals defined over the nodes of irregular topological domains such as graphs \cite{Shuman2013},\cite{Moura2014},\cite{ortega2018graph}. A cornerstone  of graph signal representation and processing is the Laplacian matrix, an  algebraic  operator that encodes the topology of the graph and enables the spectral representation and processing of signals defined on its nodes. Depending on the application and the structural properties of the graph, different Laplacian forms have been proposed in the literature, such as the combinatorial Laplacian \cite{chung1997spectral}, the signless  Laplacian \cite{cvetkovic2007}, and the signed Laplacian \cite{kunegis2010spectral}. Each of these operators is particularly suited to capture specific features of the graph. For instance, the combinatorial Laplacian is well suited for modeling smooth signals over clustered graphs, the signless Laplacian is particularly effective for bipartite structures, and the signed Laplacian naturally accounts for graphs with positive (cooperative) and negative (antagonistic) interactions. Interestingly, the  smoothness or total variation of signals observed over graphs, algebraically represented by these operators, can be characterized through their corresponding quadratic forms \cite{chung1997spectral}, \cite{kunegis2010spectral}, \cite{Matz_mag_2020}. 
However,  the eigenvectors of these Laplacians exhibit different structural  patterns. 
 Consequently, their spectral properties  lead to different signal representations, enabling graph signal processing frameworks to capture a variety of structural patterns \cite{Matz_mag_2020}, \cite{gallier2016spectral}, \cite{kunegis2010spectral}.  As a result, the choice of Laplacian plays a crucial role in determining the features that can be extracted from the data. However, most existing graph learning methods \cite{Dong16,segarra2017network,sardellitti2019graph,egilmez2017graph,de2022learning} 
 rely on prior knowledge of the appropriate Laplacian form to be used for data processing.

Recently, in \cite{zhao2025fractional} a fractional graph spectral filtering method based on parametrized matrices was proposed. In \cite{Averty}, the authors proposed families of graph representation matrices combining degree and adjacency matrices with tuned weights. In \cite{wang2020bounds} a family of graph representations called $\alpha$-Laplacian was also introduced.

Our goal in this paper is to propose a novel  parametric Laplacian operator for graph signal processing, called the  \textit{Deformed Graph Laplacian (DGL)}, tailored to both the graph structure and the signals observed on the graph. 
The deformed Laplacian was proposed in \cite{morbidi2013deformed} for consensus protocol in multi-agent systems and employed  for semisupervised learning in label prediction tasks \cite{gong2015deformed}. In \cite{saade2014spectral} it is applied to spectral clustering, showing its relation with the non-backtracking matrix and  the phase transition in an Ising model. 
The DGL  is a second-degree matrix polynomial   $\mathbf{L}_{\text{DF}}(r)$ in a real (or complex) variable $r$ 
\cite{grindrod2018deformed}, \cite{tisseur2001quadratic}. \\ In this paper we provide a novel interpretation of the deformed Laplacian  as the generator of a parametric
linear reaction-diffusion dynamics on graphs, encoding the
interplay between diffusive coupling and nodal reaction effects \cite{evans2022partial,van2023graph}.  The parameter 
$r$ governs both the intensity of diffusive coupling and the characteristics of the local dynamics. This perspective shows that the deformed Laplacian extends the combinatorial graph Laplacian by introducing a degree-dependent reaction term, thus generalizing purely diffusive dynamics to a richer class of processes. 

Furthermore, we show how the deformed Laplacian reduces to the standard Laplacian matrices, i.e. the combinatorial, the signless and the signed Laplacians, for selected values of $r$. Specifically, these standard  Laplacians are obtained for values of the parameter $r$ solving the polynomial eigenvalue problem, with associated eigenvectors in the kernel of  $\mathbf{L}_{\text{DF}}(r)$. \\
We propose a  framework for graph signal processing that addresses a key limitation of existing graph learning approaches in GSP, namely the need for prior knowledge of the underlying Laplacian form. 
Using the eigenvectors of the deformed Laplacian as parametric dictionaries for  spectral representation of graph signals, we extend GSP tools  by using this novel algebraic operator.
Then, we develop a framework that, given the graph adjacency, 
 learns from data the optimal graph model, i.e., the most appropriate Laplacian form, that  best   captures the data features.
The main contributions of this paper can be summarized as follows:
\begin{itemize}
\item[1)] The parametric deformed Laplacian is introduced as an algebraic operator that provides a unified framework encompassing different Laplacian forms through a single parameter. The operator is defined as a second-order matrix polynomial, whose spectral properties reveal that, for specific parameter values, its eigenvectors coincide with those associated with standard Laplacian operators, including the combinatorial, signless, and signed Laplacians;
\item[2)] A novel interpretation of the deformed Laplacian is provided, viewing it as a graph operator associated with a parametric reaction–diffusion dynamics. In this perspective, the operator captures the interplay between diffusive interactions across edges and local node-level reactions, extending purely diffusive dynamics to a broader class of processes;
\item[3)] A method is proposed to jointly learn the optimal Laplacian form from data and  a sparse graph signal representation, given the graph adjacency matrix. The approach overcomes a key limitation of existing graph signal processing methods, namely the requirement of prior knowledge of the Laplacian form used for spectral representation. Furthermore, it is particularly effective for graphs with heterogeneous structures, such as those combining clustered and bipartite patterns;
    \item[4)] Extensive numerical experiments demonstrate that learning the graph operator form using the deformed Laplacian improves both the accuracy of signal representation and its sparsity, particularly in scenarios where the graph topology is not well structured and exhibits a combination of clustered and bipartite patterns. The proposed approach is validated on both synthetic and real-world datasets, including financial stock networks \cite{de2022learning} and communication data from the Copenhagen Networks Study \cite{copenhagen}.
\end{itemize}
This paper is organized as follows. Section II provides a brief overview of matrix polynomials and the polynomial eigenvalue problem. In Section III we introduce the deformed Laplacian as a parametric quadratic matrix polynomial and provide a novel interpretation of the deformed Laplacian as a reaction-diffusion operator on graphs.  Section IV shows  how the combinatorial, signless, and signed Laplacians can be recovered as specific instances of the deformed Laplacian, while  Section V analyses the spectral properties of this operator.
Section VI introduces a spectral signal representation based on the eigenvectors of the deformed Laplacian. In Section VII, a learning strategy for jointly estimating the Laplacian form and a sparse signal representation is presented, together with extensive numerical results on both synthetic and real-world datasets. Finally, in Section VIII we draw some  conclusions.

\section{Matrix Polynomials: A brief overview} 
In this section, we introduce general notions of matrix polynomials which play a key role in the introduction of the Deformed Graph Laplacian (DGL). 
The theory of matrix polynomials has applications in many areas as, for instance, differential equations, systems theory, mechanics and vibrations, numerical analysis \cite{gohberg2009matrix},\cite{tisseur2001quadratic}.
A matrix polynomial is a polynomial in a complex  or real variable with matrix coefficients. 

Given a field \(\mathbb{C}\) (or \(\mathbb{R}\)), we denote by $\mathbb{C}(\lambda)$ the set of univariate polynomials in $\lambda \in \mathbb{C}$ with coefficients in $\mathbb{C}$. 
The set of square matrices of size $n$ with entries in $\mathbb{C}[\lambda]$ is denoted by $\mathbb{C}[\lambda]^{n \times n}$. For $j = 0, 1, \ldots, k$, let $\mA_j \in \mathbb{C}^{n \times n}$ be square matrices of the same size, with $\mA_k \neq 0$.
A square matrix polynomial of order $n$ and degree $k$ is the matrix-valued function \cite{gohberg2009matrix} \beq  \mP(\lambda) = \sum_{j=0}^{k} \mA_j \lambda^j \in \mathbb{C}[\lambda]^{n \times n}.\eeq
The spectrum of $\mP(\lambda)$, i.e. the set of eigenvalues of $\mP(\lambda)$, denoted by $\boldsymbol{\Sigma}(\mP)$, is defined as  \cite{gohberg2009matrix}, \cite{morbidi2012properties}
\beq \boldsymbol{\Sigma}(\mP)=\{ \lambda \in \mathbb{C} \, : \det(\mP(\lambda))=0\} \label{eq:Sigma_P}.\eeq
If the degree of  $\text{det}(\mathbf{P}(\lambda))$  is less than $kn$, $\infty$
 is said to be an eigenvalue of $\mathbf{P}(\lambda)$. 
The algebraic multiplicity of an eigenvalue $\lambda_0$ is the order of the corresponding
zero in $\text{det}(\mathbf{P}(\lambda))$ . The geometric multiplicity of $\lambda_0$ is the dimension of $\text{ker}(\mathbf{P}(\lambda))$.
Let us now introduce some properties of matrix polynomials \cite{tisseur2001quadratic}, \cite{morbidi2013deformed}.

\textit{Definition 1:}
\textit{Given the matrix polynomial $\mP(\lambda)$ of degree $k$, it  is called non regular when  $\det(\mP(\lambda)) \equiv 0$ for all values of $\lambda$ and  regular otherwise.
If  $\mP(\lambda)$ is regular, then it holds: \begin{itemize} \item[i)] if the degree $q$ of $\det(\mP(\lambda))$ is   $q=k n$, then  $\mP(\lambda)$ has $k n$ finite eigenvalues (counted with algebraic multiplicity); \item[ii)]  if the degree of $\det(\mP(\lambda))$ is $q<kn$, then $\mP(\lambda)$ has $q$ finite  eigenvalues and $kn-q$  infinite eigenvalues\footnote{The infinite eigenvalues of the matrix polynomial correspond to the zeros of the reversal polynomial  $\lambda^k \mP(\lambda^{-1})$.}.
\end{itemize}}
Let us now introduce the non-linear eigenvalue problem \cite{gohberg2009matrix}, \cite{grindrod2018deformed},\cite{bueno2011recovery}.\\ 
\textit{Definition 2:}
\textit{For regular matrix polynomials of degree $k$, the Polynomial Eigenvalue Problem
(PEP) consists of finding scalars $\lambda_0 \in \mathbb{C}$ and nonzero vectors $\bu,\bv \in \mathbb{C}^n$ satisfying
\beq \label{eq:def_eig} \mP(\lambda_0)\bu = \mathbf{0} \;\; \text{and} \;\; \bv^H\mP(\lambda_0)=\mathbf{0}.
\eeq 
The values $\lambda_0$ are known as the eigenvalues of $\mP(\lambda)$
and the associated nonzero vectors $\bu$ and $\bv$
are known as right and left eigenvectors of
$\mP(\lambda)$, respectively.}

If a square matrix polynomial $\mP(\lambda)$ is such that $\mP(\lambda)=\mP(\lambda)^H$
for all $\lambda  \in \mathbb{R}$, then
$\mP(\lambda)$ is called Hermitian \cite{gohberg2009matrix}, \cite{mehrmann2016sign}. The spectral theory of Hermitian matrix polynomials is significantly more intricate than the classical spectral theory of Hermitian matrices. For instance, in contrast to the matrix case, the eigenvalues of Hermitian matrix polynomials are not necessarily real \cite{tisseur2001quadratic}.
It is important to note that the PEP problem of the degree $1$ matrix polynomial $\mP(\lambda)=-\lambda \mA_1 + \mA_0$ is the well-known Generalized Eigenvalue Problem (GEP), expressed as 
\beq \label{eq:GEP}
\begin{array}{lll}
&(-\lambda_i \mA_1 + \mA_0) \bu_i=\mathbf{0} \medskip\\
&\bv_i^H(-\lambda_i \mA_1 + \mA_0) =\mathbf{0}.
\end{array}
\eeq
 Furthermore, assuming $\mA_1=\mI$, (\ref{eq:GEP}) reduces to the Standard Eigenvalue Problem (SEP) for the matrix $\mA_0$, namely 
 \beq \label{eq:SEP}
\begin{array}{lll}
&\mA_0 \bu_i=  \lambda_i  \bu_i \medskip\\
&\bv_i^H \mA_0 =\lambda_i \bv_i^H.
\end{array}
\eeq 

Therefore, PEPs represent a broad and important class of nonlinear eigenvalue problems. While less studied than standard (SEP) and generalized (GEP) eigenvalue problems, they encompass both as special cases and thus deserve interesting properties to be investigated \cite{gohberg2009matrix}.

\section{The Deformed  Laplacian}

Graph-based representations are powerful tools  for analyzing data, encoding pairwise relationships between the observed signals through the graph topology. 
 A graph $\mathcal{G}=(\mathcal{V},\mathcal{E})$ is composed by a set $\mathcal{V}$ of $N$ vertices  and a set $\mathcal{E}$ of edges  between nodes.
The  connectivity among the nodes  of a graph can be  captured by the adjacency matrix $\mA \in \mathbb{R}^{N\times N}$ with entries  $a_{ij} \neq 0$ if there is an edge between nodes $i$ and $j$, and $a_{ij}=0$ otherwise.
In GSP, a  graph signal $\bx \in \mathbb{R}^{N}$  on the nodes of a graph $\mathcal{G}$ is defined as a real-mapping on the set of vertices 
$\mathcal{V}$, i.e. $\bx : \mathcal{V}\rightarrow \mathbb{R}^{N}$.

We focus on undirected graphs, i.e. graphs with symmetric adjacency matrices and with no self-loops, i.e., $a_{ii}=0$ $\forall i.$ However,
 we do not limit our formulation to binary adjacency  matrices, so that we assume, unless differently stated, that the edge weights $a_{ij}$ can be real numbers, not necessarily non-negative. More specifically, a positive edge weight $a_{ij}>0$ indicates that the signals at nodes $i$ and $j$ tend to be similar (or cooperative), whereas a  negative coefficient   $a_{ij}<0$ reflects dissimilarity (or antagonism).
 To encompass both the cases of non-negative and signed coefficients, we define the adjacency matrix $\bar{\mA}$ as
 \beq \label{eq:bar_A}
\bar{\mA}=\left\{ \begin{array}{lll}
\mA_{\text{S}} & \text{if} \; a_{ij} \in \mathbb{R}, \; \;\; a_{ii}=0, \forall i,j\\
\mA & \text{if} \; a_{ij} \in \mathbb{R}^{+}, \; a_{ii}=0, \forall i,j
\end{array}
\right. .
 \eeq
 Similarly, we define  the diagonal degree matrix 
  \beq \label{eq:bar_D}
\bar{\mD}=\left\{ \begin{array}{lll}
\mD_{\text{S}} & \text{if} \; a_{ij} \in \mathbb{R}, \;  \; \; d_{ii}=\sum_{j=1}^{N}|a_{ij}| \\
\mD & \text{if} \; a_{ij} \in \mathbb{R}^{+}, \;   d_{ii}=\sum_{j=1}^{N}a_{ij}\  \end{array}
\right. .
\eeq

 We  associate with any  graph a special matrix polynomial $\mL_{\text{DF}}(r)$, named \textit{Deformed Graph Laplacian} (DGL). It  is a quadratic 
matrix polynomial \cite{tisseur2001quadratic}
\cite{grindrod2018deformed} in the variable $r \in \mathbb{C}$ \cite{gong2015deformed}, \cite{morbidi2013deformed}, \cite{morbidi2012properties}.
The DGL has been studied in \cite{morbidi2013deformed} for generalizing consensus protocol in multi-agent systems and robotics. More formally, we defined the deformed Laplacian as follows \cite{morbidi2013deformed}.\\
 \textbf{Definition 3}
\textit{Let $\mathcal{G}$ a graph with adjacency matrix $\bar{\mA}$ and degree matrix $\bar{\mD}$. For any $r \in \mathbb{C}$, the associated Deformed Graph Laplacian is the following quadratic hermitian
matrix polynomial:
\beq \label{eq:deformed_Lapl}
\mL_{\text{DF}}(r)=(\bar{\mD}-\mI) r^2 -\bar{\mA} r +\mI. 
\eeq}

It can be easily seen that, for $r \in \mathbb{R}$, $\mL_{\text{DF}}(r)$ is a symmetric matrix parameterized  with respect to $r$, in general not positive definite. Furthermore, $\mL_{\text{DF}}(r)$ is a comonic polynomial matrix, i.e. it holds $\mL_{\text{DF}}(0)=\mathbf{I}$ \cite{gohberg2009matrix}.

Considering $r \in \mathbb{R}$, the deformed Laplacian, $\mL_{\text{DF}}(r)$ is, for each $r$, a symmetric real matrix with SEP eigen-decomposition 
\beq \label{eq:eig_dec_Ldf}
\mL_{\text{DF}}(r)=\mU(r) \boldsymbol{\Lambda}(r) \mU(r)^T,
\eeq
where $\boldsymbol{\Lambda}(r)$ is  the eigenvalue diagonal matrix and  $\mU(r)$ is the $N\times N$ matrix with columns the associated eigenvectors. Note that the eigenvalues of the SEP associated with the DGL are real-valued functions of 
$r$, although generally nonlinear. In the next section, we provide a novel interpretation of the DGL as a reaction–diffusion operator on graphs.

\subsection{The Deformed Laplacian as a  Reaction-Diffusion Operator}

Diffusion processes on graphs provide a natural dynamical interpretation of the combinatorial Laplacian \cite{evans2022partial,van2023graph}. In this section, we introduce a more general class of models, namely reaction–diffusion equations \cite{evans2022partial,murray2003mathematical}, which are parabolic partial differential equations (PDEs) describing the evolution of a function under the combined effects of diffusion and local reactions, and play a central role in pattern formation theory \cite{turing1990chemical}.
Then, we provide a novel interpretation of 
the deformed Laplacian $\mL_{\mathrm{DF}}(r)$ as the generator of a parameter-dependent linear reaction-diffusion dynamics on a graph.

Let us consider  a scalar function $\phi(\bx,t)$ depending on  the spatial coordinates $\bx=(x_1,x_2) \in \mathbb{R}^2$ and time $t$.
The continuous-time diffusion-reaction equation can be written as \cite{evans2022partial}, \cite{murray2003mathematical}: 
\beq \label{eq:cont_diff_reac}
\frac{\partial \phi(\bx,t)}{\partial t}=\underbrace{\eta \nabla^2 \phi(\bx,t)}_{\text{diffusion}} + \underbrace{h(\bx,t)}_{\text{reaction}} 
\eeq
where $\eta>0$ is  the diffusion coefficient and $h(\bx,t)$  accounts for  local reaction terms.
Note that in the absence of the reaction term, equation (\ref{eq:cont_diff_reac}) reduces to the classical diffusion equation.
By spatial discretization of the continuous equation in (\ref{eq:cont_diff_reac}) \cite{nakao2010turing}, we  obtain the following graph-based dynamics
\beq \label{eq:diff_reac}
\frac{d \boldsymbol{\phi}(t)}{d t}=-\eta \mL\boldsymbol{\phi}(t)+ \boldsymbol{h}(t)
\eeq
where $\mL=\mD-\mA$ is the combinatorial Laplacian and the vectors $\boldsymbol{\phi}(t)$ and $\boldsymbol{h}(t)$ collect the values $\phi_i(t)$ and $h_i(t)$ at each node $i$, respectively.
Note that the first term represents the discrete counterpart of diffusion, arising from linear interactions along the edges and enforcing a local balance at each node, whereas the second term accounts for local reaction dynamics.\\
Inspired by  (\ref{eq:diff_reac}), we propose a novel interpretation of diffusion-reaction dynamics on graph leading to a parametric Laplacian, namely the deformed Laplacian, which includes local reaction effects. In this framework, the evolution at each node is governed by the interplay between diffusive coupling and node-dependent production or dissipation terms.
Specifically, we assume in (\ref{eq:diff_reac})
that the local reaction is modeled as 
\beq
h_i(t)=\phi_i(t)[(1-d_{ii})r^2+d_{ii} r-1]
\eeq
that is proportional to  $\phi_i(t)$
through the second-order polynomial $p(r)=(1-d_{ii})r^2+d_{ii} r-1$.
Under this assumption, and setting $\eta=1$, (\ref{eq:diff_reac}) can be written for each node $i=1,\ldots,N$ as 
\beq
\frac{d {\phi}_i(t)}{d t}=- r \sum_{j=1}^{N} a_{i,j}(\phi_i(t)-\phi_j(t))+\phi_i(t)[(1-d_{ii})r^2+d_{ii} r-1] \; 
\eeq
where the first term corresponds to a diffusive coupling across adjacent nodes,  while the second term represents a local reaction  explicitly depending on the node degree.
Collecting all node dynamics in vector form, we obtain
\beq \label{eq:diff_reacdef}
\frac{d \boldsymbol{\phi}(t)}{d t}=-\mL_{\text{DF}}(r)\boldsymbol{\phi}(t)
\eeq
which shows that the dynamics are governed by the deformed Laplacian
 \beq \mL_{\mathrm{DF}}(r) = (\mD - \mI)r^2 - \mA r + \mI. \eeq In \cite{morbidi2012properties}, the deformed consensus protocol is analyzed by investigating the convergence properties of the dynamical system in \eqref{eq:diff_reacdef}, without providing an interpretation of the associated operator as a reaction-diffusion operator.

In contrast, our interpretation shows that 
$\mL_{\mathrm{DF}}(r)$ can be viewed as a parameter-dependent reaction-diffusion operator on the graph, where the parameter 
$r$ modulates both the strength of the diffusive coupling and the nature of the local dynamics.
In this sense, 
$\mL_{\mathrm{DF}}(r)$ generalizes the standard graph Laplacian by incorporating a degree-dependent reaction term, thereby extending purely diffusive dynamics to a broader class of processes. In particular, the parameter 
$r$ enables a continuous transition between regimes dominated by diffusion and regimes characterized by local amplification or dissipation.

\section{The Deformed Graph Laplacian: A Unified Laplacian form}

A key property of the deformed Laplacian is that it is a parametric matrix which, for specific values of $r$, reduces to other well-known Laplacian forms  commonly used in GSP. In fact, there are several variants of the graph Laplacian, such as the combinatorial \cite{chung1997spectral}, the signless \cite{cvetkovic2007} and the signed Laplacians \cite{kunegis2010spectral}. Each of these forms is able to reveal different properties of the graph by capturing different features from data. In the following, we will show how these matrices can be derived as special cases of the deformed Laplacian $\mL_{\text{DF}}(r)$.

\subsection{The combinatorial graph Laplacian
}
Let us consider an undirected graph whose incidence matrix $\mA$ has non-negative entries $a_{ij}\geq 0$.
The combinatorial Laplacian matrix \cite{MERRIS94},\cite{chung1997spectral}, $\mL=\mD-\mA$,  can be derived from (\ref{eq:deformed_Lapl}) by setting $r=1$. Thus, we get 
\beq \mL_{\text{DF}}(1)\equiv \mL=\mD-\mA. \eeq
which is a symmetric, semidefinite positive matrix \cite{mohar1991laplacian}.  For connected graphs the first eigenvalue $\lambda_1(\mL)$ is equal to zero with associated eigenvector the constant vector $\mathbf{1}$. 
The algebraic connectivity of $\lambda_1(\mL)$ is equal to the number of connected component of $\mathcal{G}$.
 The eigenvector associated with the zero eigenvalue is a constant vector, and it holds
$
\mL_{\text{DF}}(1) \mathbf{1}=\mathbf{0}.
$
Then, according to  (\ref{eq:def_eig}),  $r=1$ is an eigenvalue of the polynomial matrix $\mL_{\text{DF}}(r)$ (since $\text{det}(\mL_{\text{DF}}(1))=0$) with associated constant eigenvector.
In a graph with clusters the 
 eigenvectors associated with the smallest eigenvalues of $\mL$ are smooth within clusters.
 The quadratic form built on the combinatorial Laplacian provides a measure of the quadratic total variation of a graph signal $\bx \in \mathbb{R}^N$ and is given by 
\beq \label{eq:quad_form_L}
\text{QTV}(\bx)=\bx^T \mL \bx=\ds \sum_{i,j=1}^{N} a_{ij}(x_i-x_j)^2.
\eeq
Then, the eigenvectors of $\mL$ are the orthonormal bases minimizing the quadratic total  variation  with orthonormal constraints and  they are useful bases to
identify  clusters.

To illustrate such behavior, we consider in Fig. \ref{fig:fig_kar} the network of Zachary’s Karate Club, which encodes social interactions between members of a university karate club \cite{zachary1977information},\cite{nr}.  The graph is composed of $N=34$ nodes representing individuals in the club, while links represent social interactions. This network is known to split into two main communities as a result of an internal conflict. In Fig. \ref{fig:fig_kar} we plot  on the nodes the values of the  Fiedler eigenvector, i.e., the second eigenvector, of the graph Laplacian $\mL$.
 It can be noticed that the eigenvector assumes similar values on the nodes within three clusters, suggesting the presence of subgroups within the two principal communities.
\begin{figure}[t] 
\centering
	\includegraphics[width=8.0cm,height=5.0cm]{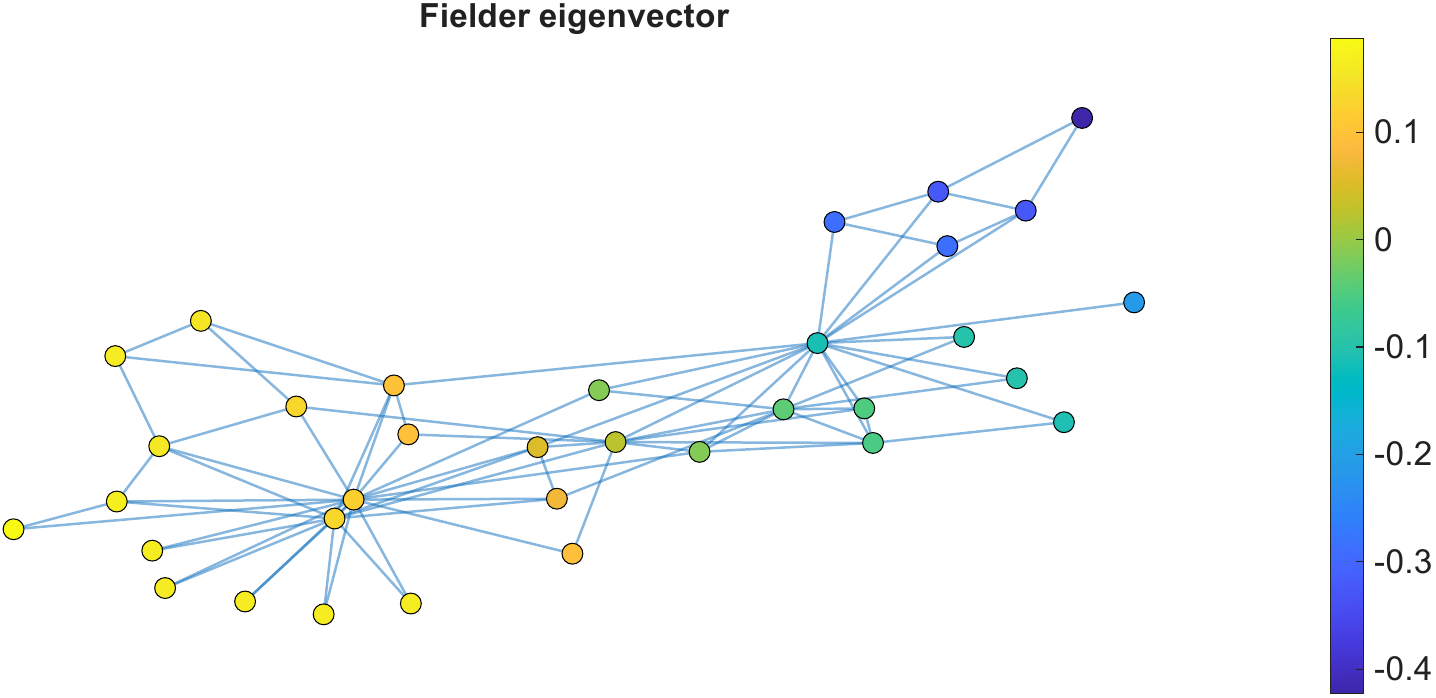}
	
    \caption{Second eigenvector  of  $\mL$ on the Zachary's karate club network.}\label{fig:fig_kar}
\end{figure}

\subsection{The signless graph Laplacian} 
A further  relevant, albeit less known, operator for graph representation is the \textit{signless} Laplacian associated with a graph $\mathcal{G}$.
Assuming non-negative adjacency matrices, the \textit{signless} Laplacian of a graph $\mathcal{G}$ is defined as  \cite{cvetkovic2007} \beq \mL_{\text{SL}}=\mD+\mA. \eeq Setting $r=-1$ in (\ref{eq:deformed_Lapl}) the deformed Laplacian reduces to the signless Laplacian, i.e.,  
\beq 
\mL_{\text{SL}} \equiv \mL_{\text{DF}}(-1)=\mD+\mA.
\eeq
The signless Laplacian is  more suitable   for identifying the bipartiteness  of  graphs. A graph $\mathcal{G}$ is bipartite if the set of its vertices $\mathcal{V}$ can be divided into two disjoint sets $\mathcal{V}_1,\mathcal{V}_2$ such that each edge in $\mathcal{E}$ connects a node in $\mathcal{V}_1$ to one in $\mathcal{V}_2$. The matrix $\mL_{\text{SL}}$ is symmetric, semidefinite positive, and it is positive definite if and only if the underlying graph is not bipartite \cite{brouwer2011spectra}. 
Then, the following properties hold for the signless Laplacian \cite{cvetkovic2007}[Prop. 2.1 items i) and ii)].
\begin{proposition} 
\textit{For any graph, the following statements hold:
\begin{itemize}
\item[i)] The least eigenvalue of the signless Laplacian of a connected graph is equal to $0$ 
 if and only if the graph is bipartite. In this case $0$ is a simple eigenvalue;
 \item[ii)]  The multiplicity of the eigenvalue $0$ of the signless Laplacian is equal
to the number of bipartite components;
\item[iii)] For bipartite graphs,  it holds
$\text{det}(\mL_{\text{DF}}(-1))=0$, 
then $r=-1$ is an eigenvalue of the matrix polynomial $\mL_{\text{DF}}(r)$ with  multiplicity equal to the number of bipartite components in the graph.  
\end{itemize}}
 \end{proposition}
 \begin{proof}The proofs of points i) and ii) follow similar arguments as those in Proposition  $2.1$ of \cite{cvetkovic2007}, while the proof of point iii) follows from the identity $\text{det}(\mL_{\text{SL}})=\text{det}(\mL_{\text{DF}}(-1))=0$ and from the definition of matrix polynomial eigenvalues in (\ref{eq:Sigma_P}).\end{proof}
 
The quadratic form associated with the signless Laplacian naturally induces a notion of total variation of a graph signal $\bx \in \mathbb{R}^N$. In particular,  we get \cite{kunegis2015exploiting}:
\beq \text{QTV}_{\text{SL}}(\bx)=\bx^T \mL_{\text{SL}} \bx=\sum_{i,j=1}^{N} a_{ij} (x_i+x_j)^2 \eeq
whose minimum zero value is achieved when $x_i=-x_j$ for all connected nodes $i,j$.\\
\begin{figure}[t] 
\centering
	\includegraphics[width=8.0cm,height=4.7cm]{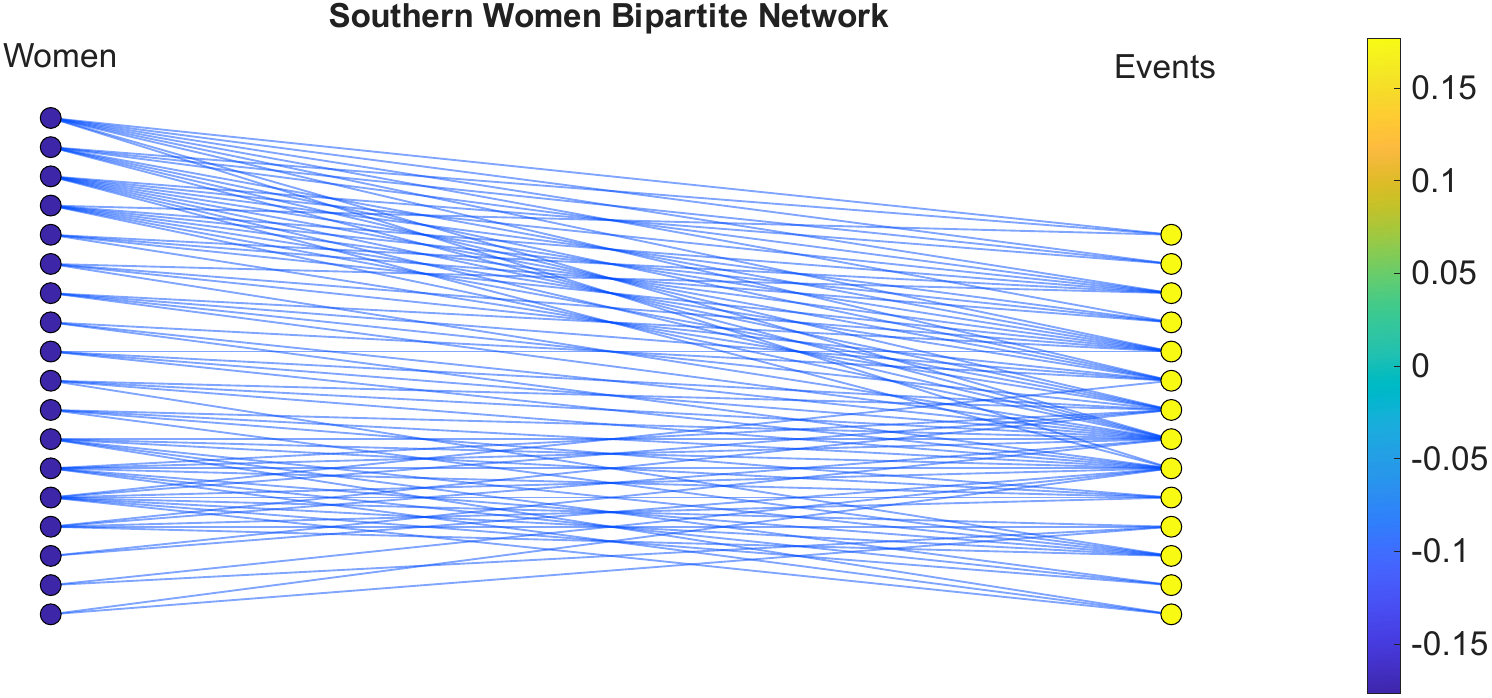}
	
    \caption{Eigenvector  of  $\mL_{\text{SL}}$ associate with the zero eigenvalue on the Southern Women bipartite network.}\label{fig:fig_bip}
\end{figure}
The bipartite structure of a graph can be recovered by examining the sign pattern of the eigenvector associated with the smallest zero eigenvalue of $\mL_{{\text{SL}}}$ 
\cite{goldberg2014sign},\cite{kunegis2015exploiting}. It is positive and constant on $\mathcal{V}_1$, and negative and
constant on $\mathcal{V}_2$. Then,   the sign pattern of the eigenvector associated with the zero eigenvalue reveals the bipartite components of the graph by assuming values of opposite sign over nodes belonging to different bi-partitions.

As an example, we consider the bipartite Southern Women network     \cite{womendataset}, \cite{davis1941deep}, as provided by the NetworkX library \cite{hagberg2007exploring}. The graph represents social participation patterns among a group of women in Natchez, Mississippi. The dataset encodes affiliation relationships between 18 women and 14 social events, forming a bipartite network, illustrated in Fig. \ref{fig:fig_bip}. An edge connects a woman to an event if she attended that event. 
 The signless Laplacian is semidefinite positive and has a single zero eigenvalue.
In Fig.  \ref{fig:fig_bip}  we report on the nodes the values of the first eigenvector of $\mL_{\text{SL}}$ for the bipartite graph. It can be noticed as the eigenvector assumes  negative  and constant values over the nodes in the women set and positive, constant values on those in the event set.\\ 
In \cite{goldberg2014sign}, it is shown that 
structural bipartiteness properties can be derived
also in the case where the smallest eigenvalue of $\mL_{\text{SL}}$ is small but nonzero.


\subsection{The Signed Laplacian} 
Let us now consider  undirected weighted graphs, where an edge with  positive weight denotes similarity between the signals observed over the vertices of the edge.  
To indicate dissimilarity or distance between data, 
in many applications the edges are labeled with negative coefficients. Such graphs for which the adjacency matrix has negative and
positive entries are called signed graphs. The associated  graph structure is captured by  the so called \textit{signed Laplacian}  \cite{kunegis2010spectral}, \cite{Matz_mag_2020},
\cite{gallier2016spectral}, defined as 
\beq 
\mL_{\text{S}}={\mD}_{\text{S}}-\mA_{\text{S}}.
\eeq
The deformed Laplacian matrix reduces to the signed Laplacian by setting $r=1$ in (\ref{eq:deformed_Lapl}), namely,
\beq 
\mL_{\text{S}} \equiv \mL_{\text{DF}}(1)={\mD}_{\text{S}}-\mA_{\text{S}}.
\eeq
Let us now define  a balanced graph. A graph is balanced if it can be split into two disjoint vertex sets $\mathcal{V}_1$ and $\mathcal{V}_2$ such that all edges between $\mathcal{V}_1$ and $\mathcal{V}_2$ are negative, whereas all edges within  $\mathcal{V}_1$ and $\mathcal{V}_2$ are positive \cite{Matz_mag_2020}. 
Then, the following properties hold.
\begin{proposition} 
\textit{For any graph, the following statements hold:
\begin{itemize}
\item[i)] The signed Laplacian is a strictly definite positive matrix and it is semidefinite positive if the graph is balanced;
\item[ii)] For balanced graph, $r=1$ is an eigenvalue of the matrix polynomial, i.e., $\text{det}(\mL_{\text{DF}}(1))=0$.
\end{itemize}}
 \end{proposition}
\begin{proof}
    The proof of point i) is provided in   \cite{kunegis2010spectral}. The  statement in point ii) follows from the identity $\mL_{\text{S}} = \mL_{\text{DF}}(1)$ and  from the fact that, for a balanced graph, $\text{det}(\mL_{\text{S}})=0$. This implies that $\text{det}(\mL_{\text{DF}}(1))=0$,  and hence $1$ is an eigenvalue of the deformed Laplacian.
\end{proof}

The quadratic form associated with the signed Laplacian can be interpreted  as a measure of graph signal smoothness \cite{Matz_mag_2020}, so that we can write 
\beq \text{QTV}_{\text{S}}(\bx)=\bx^T \mL_{\text{S}} \bx= \frac{1}{2} \ds \sum_{i,j=1}^{N} |a_{ij}| (x_i-\text{sign}(a_{ij})x_j)^2. \eeq
Note that $\text{QTV}_{\text{S}}(\bx)$ is small when  node signal values are similar across positively weighted edges (i.e., $x_i \approx  x_j$ for $a_{ij}> 0$)  and dissimilar across negatively weighted edges (i.e., $x_i \approx - x_j$ for $a_{ij} < 0$).
The  smallest signed Laplacian eigenvalue measures the level of balance in a graph: the larger its value, the more unbalanced the network. Then, in a balanced graph, the sign pattern of the eigenvector corresponding to the smallest (zero) eigenvalue identifies disjoint vertex sets.

As an illustrative example,  Fig. \ref{fig:fig_sign}  shows a signed balanced graph with  two sets of nodes, each composed by $10$ nodes.
The  green links within each cluster are positively weighted, while the red edges are negatively weighted. On the nodes we represent the eigenvector associated with the first zero eigenvalue of $\mL_{\text{S}}$, which assumes opposite values $\pm 0.22$. Notably, the sign pattern of this eigenvector clearly identifies the two disjoint node sets.
\begin{figure}[t] 
\centering	\includegraphics[width=8.0cm,height=4.7cm]{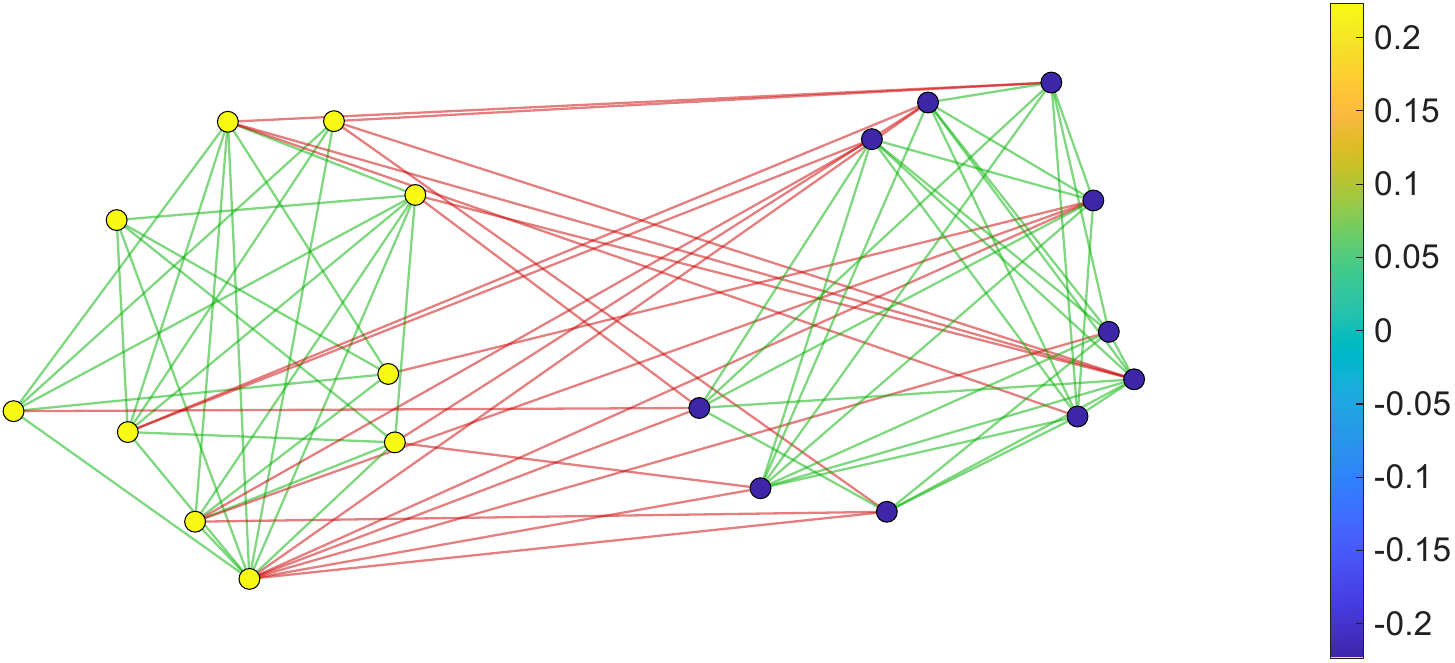}	
    \caption{First eigenvector  of  $\mL_{\text{S}}$ associate with the zero eigenvalue.}\label{fig:fig_sign}
\end{figure}

\section{Spectral properties of the Deformed Laplacian}
In this section, we introduce some basic spectral properties of the deformed graph Laplacian as presented in \cite{grindrod2018deformed}[Prop. 4.4] and extended here to positively weighted graphs.
\begin{proposition} \textit{Consider the deformed Laplacian defined in \eqref{eq:deformed_Lapl}, and assume that the graph has positive edge weights. Then, the following properties hold:
\begin{itemize}
    \item[a)] $\mL_{\text{DF}}(r)$ is a regular matrix polynomial, and $0$ is never an eigenvalue of $\mL_{\text{DF}}(r)$;
     \item[b)] $1$ is always an eigenvalue of $\mL_{\text{DF}}(r)$, with geometric multiplicity equal to the
number of connected components of the graph of $\mA$;
\item[c)] The geometric multiplicity of $\infty$ as an eigenvalue of $\mL_{\text{DF}}(r)$ is equal to the number  vertices of degree $1$, in the graph of $\mA$. Then, $\infty$ is
an eigenvalue of $\mL_{\text{DF}}(r)$ if and only if the graph of $\mA$ has at least one vertex of degree $1$;
\item[d)] $-1$ is an eigenvalue of $\mL_{\text{DF}}(r)$ if and only if the graph of $\mA$ has at least one
bipartite component. In this case, the geometric multiplicity of $-1$ is equal to
the number of bipartite components of the graph of $\mA$.
\end{itemize}}
\end{proposition}
\begin{proof}
The proof of points a), b) and d) follows the same arguments of the proof of  Prop. 4.4  in \cite{grindrod2018deformed}. The statement in  c) relies on the fact that the geometric multiplicity of $\infty$ is equal to $\text{dim} (\text{ker}(\mD-\mathbf{I}))$.  Consequently, $\infty$ is an eigenvalue if and only if  there exists at least one node with weighted degree $d_{ii}=1$ and  its multiplicity is given by the number of vertices with weighted degree equal to $1$.
\end{proof}\\
Considering binary symmetric adjacency matrices from Theorem 4.8 in \cite{grindrod2018deformed}, it follows that for each finite eigenvalue $\lambda \in \mathbb{C}$ of $\mL_{\text{DF}}(r)$, we get $|\lambda|<1$.\\
Assuming that the graph is signed, we can easily prove the following proposition.
\begin{proposition}
    \textit{Consider the deformed Laplacian in (\ref{eq:deformed_Lapl}) and assume that the adjacency matrix is signed. Then, the following properties hold:
\begin{itemize}
    \item[a)] $\mL_{\text{DF}}(r)$ is a regular matrix polynomial, and $0$ is never an eigenvalue of $\mL_{\text{DF}}(r)$;
     \item[b)] The graph is balanced if and only if  $\lambda=1$ is an eigenvalue of the deformed Laplacian. The multiplicity of $\lambda=1$ is the number of balanced components of the graph.
    \end{itemize}}
    \end{proposition}
    \begin{proof}
    The proof of point a) is straightforward, since evaluating $
\mL_{\text{DF}}(r)=(\bar{\mD}-\mI) r^2 -\bar{\mA} r +\mI$  at $r=0$ gives $
\mL_{\text{DF}}(0)=\mI$, and then, $\text{det}(\mL_{\text{DF}}(0))=1$. This implies from (\ref{eq:Sigma_P}), that $0$ is never an eigenvalue of $\mL_{\text{DF}}$. \\
To prove point b), assume that the graph is balanced. Then there exists a vector $\bu$ with entries $\pm 1$ such that $\mL_{\text{S}} \bu=\mathbf{0}$.
Since $\mL_{\text{DF}}(1)=\mL_{\text{S}}$, it follows that $\mL_{\text{DF}}(1)\bu=\mathbf{0}$, and 
$\text{det}(\mL_{\text{DF}}(1))=\text{det}(\mL_{\text{S}})=0$.
Therefore, $\lambda=1$ is also an eigenvector of the  deformed Laplacian. Conversely, suppose that there exists a nonzero vector
$\bu$  such that $\mL_{\text{DF}}(1)\bu=\mathbf{0}$. Then,  $\mL_{\text{S}}\bu=\mathbf{0}$, which implies that the graph is balanced.
Finally, the multiplicity of $\lambda=1$ as an eigenvalue of $\mL_{\mathrm{DF}}$ is equal to
\beq
\dim( \ker\bigl(\mL_{\mathrm{DF}}(1)\bigr))
= \dim (\ker(\mL_{\text{S}})).
\eeq
Since the dimension of the kernel of the signed Laplacian equals the number of balanced connected components of the graph \cite{kunegis2010spectral}, it follows that the multiplicity of $\lambda=1$ is exactly the number of balanced connected components.
    \end{proof}


\section{Spectral Signal representation via deformed graph Laplacians} 
 
In graph signal processing (GSP), spectral methods \cite{chung1997spectral,Moura2014} are built upon a graph shift operator (GSO), with the combinatorial graph Laplacian $\mL$
 being a common choice.
Given the eigendecomposition $\mL=\mU \boldsymbol{\Lambda}\mU^T$, it is well known that a suitable basis to represent signals defined over the vertices of a graph is given by the eigenvectors $\mU$. Specifically,
the Graph Fourier Transform (GFT) ${\bs}$ of a graph signal $\bx$ is defined as the projection of $\bx$ on the space spanned by the eigenvectors of $\mL$, i.e. \cite{Shuman2013}, \cite{Moura2014}
\beq {\bs}=\mU^T \bx. \eeq
 Therefore, the graph signal admits the spectral representation 
$\bx=\mU {\bs}.$
A graph signal is  band-limited with bandwidth $K$ if  its GFT ${\bs}$ is a $K$-sparse vector. In this case,  the signal can be represented as  a linear combination of  $K$ eigenvectors as 
\beq
\bx=\mU_{\mathcal{K}} \bs_{\mathcal{K}} 
\eeq
where $\mU_{\mathcal{K}}$ contains as columns the eigenvectors indexed by the set $\mathcal{K}$ and 
$\bs_{\mathcal{K}}$ collects the corresponding GFT coefficients, with $| \mathcal{K}|=K$.

Spectral graph theory represents  a suitable  framework to extract features from  graphs via the eigenvectors of the Laplacian matrix. However, the signal representation strongly depends on the Laplacian form adopted to analyze the observed data. This, in turn, requires prior knowledge of the graph structure:  we use the combinatorial Laplacian when the graph contains clusters, the signless Laplacian if the graph is bipartite and the signed Laplacian when the edge weights are signed and there are antagonist sets. However, such prior information is often unavailable, or  graphs can be less structured, for instance  containing  both clusters and bipartite components. In such cases, we may learn from data the appropriate Laplacian form using the deformed graph Laplacian. 
Therefore, the knowledge of $\mL_{\text{DF}}(r)$  enables the representation of  node signals in terms of its eigenvectors. We define the Deformed GFT of $\bx$ as the projection of $\bx$ onto the eigenvectors $\mU(r)$ of the deformed 
Laplacian $\mL_{\text{DF}}(r)$.
Therefore, the Deformed GFT (DGFT) ${\bs}(r)$ of the node signal $\bx$ is defined as \beq {\bs}(r)=\mU(r)^T \bx \eeq
while the Inverse Deformed GFT is given by 
${\bx}(r)=\mU(r) {\bs}(r)$. Then, the signal ${\bx}$ is $K$-bandlimited on the graph modeled by the DGL if it  admits the sparse spectral representation: 
\beq
{\bx}(r)=\mU_{\mathcal{K}}(r) {\bs}_{\mathcal{K}}(r).
\eeq

As numerical example, we consider the graph signal $\bx$ represented on  the node of the graph shown in Fig.~\ref{fig:fig_graph_mix}. The graph consists of $N=11$ nodes and exhibits a mixed structure: it is quasi-bipartite, with two main components comprising nodes $1-5$ and $6-11$, respectively. Moreover, two prominent clusters can be identified, namely $\{1,2,3,6,7,8\}$ and the cluster with the complementary set of nodes.

In Fig.~\ref{fig:fig_nmse}, we report the normalized mean squared error
$\mathrm{NMSE}(r)=\frac{\|\bx - \hat{\bx}(r)\|}{\|\bx\|}$ versus the parameter $r$.
The reconstructed signal $\hat{\bx}(r)$ is obtained by projecting  $\bx$ onto the subspace spanned by the  eigenvectors $\mU(r)$ and retaining the $K=4$ largest DGFT coefficients $\bs_{\mathcal{K}}(r)$, with $|\mathcal{K}|=4$. Accordingly, we recover the band-limited signal
$\hat{\bx}(r)=\mU_{\mathcal{K}}(r)\,\bs_{\mathcal{K}}(r)$.
From Fig.~\ref{fig:fig_nmse}, it can be observed that the NMSE attains its minimum at $r=0.1$. This indicates that the spectral bases corresponding to $r=-1$ (combinatorial Laplacian) and $r=1$ (signless Laplacian) yield a higher reconstruction error, and therefore do not provide the most efficient representation of the smooth, bandlimited graph signal in this example.
The deformed Laplacian, instead, achieves a lower reconstruction error by adapting the spectral basis to the observed graph signal through the parameter $r$.
\begin{figure}[t] 
\centering	\includegraphics[width=8.0cm,height=5.5cm]{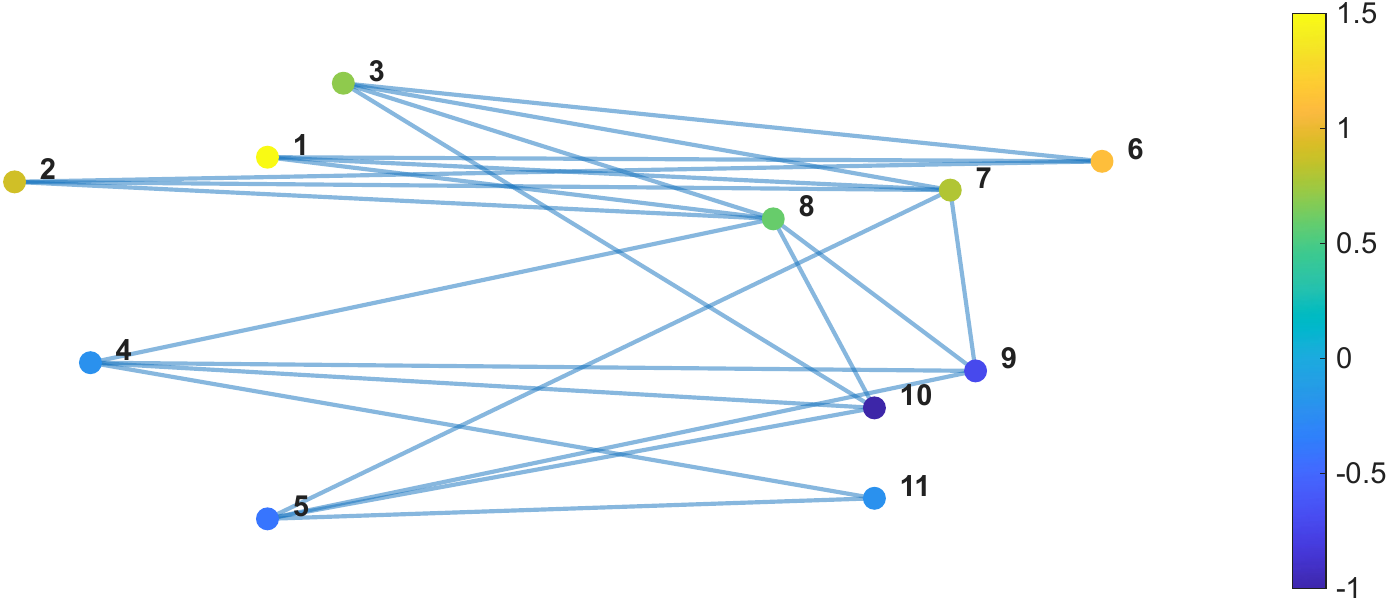}	
    \caption{Node signal on a quasi-bipartite, clustered graph.}\label{fig:fig_graph_mix}
\end{figure}
\begin{figure}[t] 
\centering	\includegraphics[width=8.8cm,height=5cm]{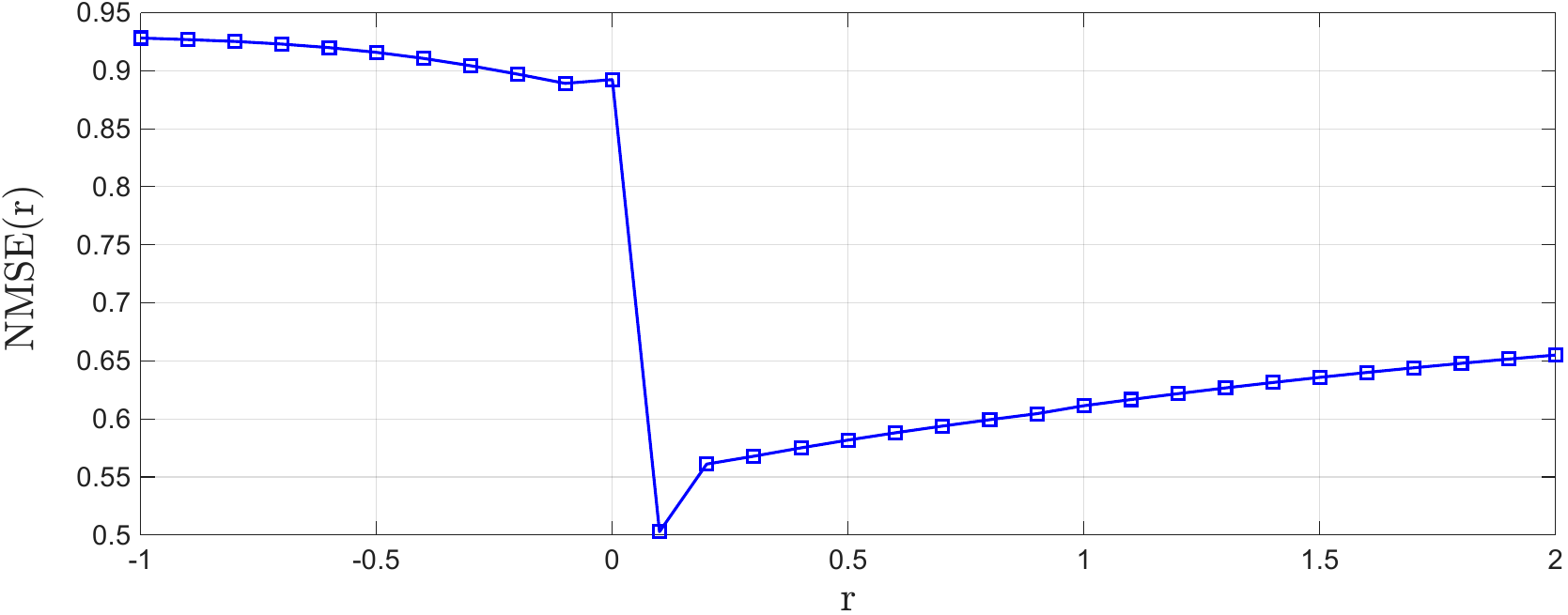}	
    \caption{Signal recovering error versus the parameter $r$ with  $K=4$.}\label{fig:fig_nmse}
\end{figure}



\section{Joint learning of Laplacian form and graph signals}

The deformed Laplacian is a parametric graph operator that enables an adaptive spectral representation of graph signals through a data-driven selection of the parameter $r$.
In this section we propose a strategy for jointly learning the deformed Laplacian form and a sparse spectral representation of the signals. Assume that a set of $M$ graph signals $\bx(i)$ are observed and stacked as columns of the  matrix $\mX=[\bx(1), \ldots, \bx(M)] \in \mathbb{R}^{N\times M}$. Let us define with $\mS=[\bs(1), \ldots, \bs(M)]$  the matrix collecting  the corresponding  DGFT vector $\bs(i)$. \\ Our goal is to find the optimal trade-off between the signal smoothness and the data-fitting error in order to jointly infer the graph Laplacian form and the sparse signal representation.
Hence, we formulate the following SDP problem:
\beq \nonumber
\begin{array}{lll}
 \underset{\substack{\mS \in \mathbb{R}^{N}, r \in \mathbb{R}\\ \mU(r) \in \mathbb{R}^{N\times N }}}{\min}  
 \!   \!\!\!\!\!\!\!\! \!\!\!\!\!\!\!\! &  \!\!f(r) \!\!= \!(1-\gamma)\text{tr}(\mX^T \mL_{\text{DF}}(r) \,\mX) \!+  \!\gamma \! \parallel \mX -\mU(r) \mS \parallel_{F}^{2} \medskip\\ 
 \qquad\text{s.t.} &  \text{(a)}\; \mL_{\text{DF}}(r)=\mU(r) \boldsymbol{\Lambda}(r)\mU(r)^T, \; \mU(r)^T\mU(r)=\mathbf{I},\medskip\\    & \text{(b)}\; \Lambda_{ij}(r)=0, \forall i \neq j, \, \boldsymbol{\Lambda}(r)\succeq \mathbf{0},\medskip\\
     & \text{(c)}\; \parallel \bs(i) \parallel_{0}= K, \; i=1,\ldots,M \qquad \qquad (\mathcal{P})
\end{array}
\eeq
where: i) the first term in the objective function promotes smoothness in the representation of the observed signal, while  the second term ensures that the graph signal representation minimizes the data-fitting error; ii)  constraints (a) and (b) impose that the deformed Laplacian admits an eigendecomposition leading to a positive semidefinite matrix; finally, iii)  constraint (c) enforces that the graph signals  are $K$-bandlimited. The coefficient $\gamma \in [0,1]$ is introduced to control the trade-off between smoothness and data-fitting error.

Problem $\mathcal{P}$ is non-convex  and to solve it  we propose an iterative efficient sub-optimal algorithm finding a solution  via a line search over the $r$ values.
Since many  Laplacians of interest arise at $r=1$ and $r=-1$, we restrict the range of values of $r$ within the real interval $[-\alpha, \alpha]$ with $\alpha\geq 1$.
Hence, for $0 \leq \gamma\leq 1$ we iteratively solve problem $\mathcal{P}$ by fixing $r$. Given the parameter $r$,  the eigenvectors $\mathbf{U}(r)$ of the associated deformed Laplacian can be calculated, and $\mathcal{P}$ admits a closed form solution  obtained by considering the projection $\hat{\bs}(i)=\mU(r)^T \bx(i)$ and keeping the $K$ largest coefficients in magnitude.

Specifically, the algorithm works as illustrated in Alg. \ref{algorithm:Alg_1}. Given as input the graph signal matrix $\mX$, the adjacency matrix $\mA$,  and an initial value $r_0$, the deformed Laplacian $\mathbf{L}_{\text{DF}}(r_0)$ is computed in Step $1$, and its eigendecomposition is derived in Step $2$. Hence, in Step $3$, we check whether the corresponding DGL is semidefinite positive. If $\mathbf{L}_{\text{DF}}(r_0)\succeq \mathbf{0}$, the optimal graph signal solving  $\mathcal{P}$ is derived by retaining the $K$ largest coefficients in magnitude of  $\hat{\bs}(i)=\mU(r_0)^T \bx(i)$. Then, the corresponding value of the objective function $f(r_0)$ is evaluated and compared with the current minimum $f_{min}$. This procedure is repeated by varying $r$ up to a maximum value $\alpha$, and the final optimal value $r^{\star}$ is obtained. Then, the sparse signal matrix  is derived as $\hat{\mX}=\mathbf{U}_{\mathcal{K}}(r^{\star})\mS_{\mathcal{K}}$. \\
In the next sections, we assess the effectiveness of the proposed method on both synthetic and real datasets.

\begin{algorithm}[t]

    \quad  {\textbf{Data:}} $\mX$, $\mA$, $\mD$.  {Set}  $n=0$, ${r}_0=-\alpha=-1$,
    $0\leq \epsilon \leq 1$ {and}\\ \makebox[1.2cm][l]{} $0 \leq \beta \ll 1$, $\text{f}_{min}=g$,    
     with $g \gg 1$

     \quad {\textbf{Output:}} $r^{\star}$, $\mL_{\text{DF}}(r^{\star})$, $x^{\star}=\mU(r^{\star}) \bs^{\star}$

    \quad  {Repeat} {until} $|r_n -\alpha| \leq \epsilon$

      \quad \quad  1. {Compute} $\mL_{\text{DF}}(r_n)$

       \quad \quad  2. {Eigendecomposition}  $\mL_{\text{DF}}(r_n)=\mU(r) \boldsymbol{\Lambda}(r) \mU(r)^T$ 
      
     \quad \quad   3. {if} $\mL_{\text{DF}}(r_n)\succeq \mathbf{0}$ 
      
     \qquad \quad  - solve $\mathcal{P}$ by fixing $r_n$ and  find\\ \makebox[1.2cm][l]{}  $\hat{\bx}_n(i)=\mathbf{U}_{\mathcal{K}}(r_n)\bs_{\mathcal{K},n}(i)$, $i=1,\ldots,M$  
     
     \qquad \quad  - {if} $f(r_n) \leq \text{f}_{min}$ 
     
   \qquad \quad  - $\text{f}_{min}=f(r_n)$,  $r^{\star}=r_n$

 \quad \quad end, end 
 
\quad \quad  4. $n=n+1$, $r_n=r_{n-1}+n \beta$  


\quad \quad  end
   \quad 
   \caption{}
 \label{algorithm:Alg_1}
\end{algorithm}

\begin{figure*}[t] 
\centering
	\includegraphics[width=15cm,height=5.0cm]{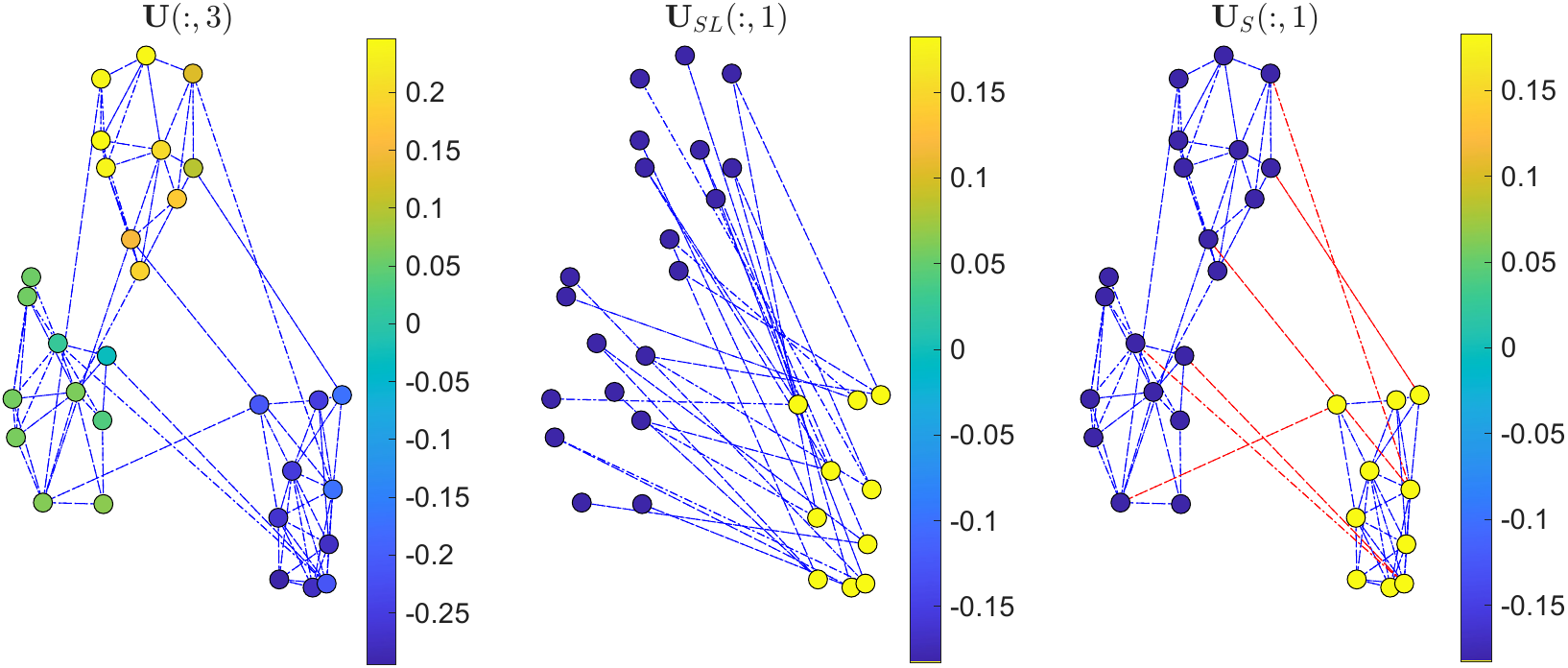}
	
    \caption{Eigenvectors of  $\mL$ (left), $\mL_{\text{SL}}$ (middle) and   $\mL_{\text{S}}$ (right).}\label{fig:fig_1}
\end{figure*}

\subsection{Synthetic Data Experiments}

As  numerical tests to evaluate the performance of the proposed framework, we first construct synthetic datasets. Specifically, we generate  the three graphs illustrated in Fig. \ref{fig:fig_1}, namely  a clustered graph, a bipartite, and a (balanced) signed graph,  covering representative graph structures of interest in many applications. These graphs are algebraically characterized by the combinatorial, signless, and signed Laplacians, respectively.

In the left plot of Fig. \ref{fig:fig_1}, we represent as nodal signal the values of the third eigenvector   of $\mL$, which is smooth within clusters. In the middle plot, we report the first  eigenvector of $\mathbf{L}_{\text{SL}}$ which takes opposite values across the two bipartite components. Finally, in the right plot, assigning  adjacency entries $-1$ to the red links and $1$ to the remaining active edges, we illustrate the first eigenvector of  $\mathbf{L}_{\text{S}}$ that is smooth within each component of the graph.
For each graph, $M$ random smooth signals are generated as  $\bx(i)=\mU_{\mathcal{K}}\, {\bs}_{\mathcal{K}}(i)$, $i=1,\ldots,M$,   using as columns of  $\mU_{\mathcal{K}}$ the first $K=3$ eigenvectors of  the combinatorial,  signless  and signed
Laplacians, respectively. 
Then, by running Algorithm $1$,  we illustrate in Fig. \ref{fig:fig_2}  the optimal value  $\bar{r}^{\star}$ as a function of  $\gamma$, averaged over $M=50$ random graph signal realizations. It can be observed that,  for both the clustered  and the signed graphs,  as $\gamma$ tends to $1$, the averaged optimal parameter value approaches $\bar{r}^{\star}=1$, since in this case the data-fitting error is the term prevailing in the optimization strategy. Conversely, as $\gamma$ tends to $0$, the method tends to select a Laplacian form that promotes smoothness of the observed signals. From Fig. \ref{fig:fig_2}, we also observe that, for the bipartite graph, the method is robust with respect to $\gamma$, as we get a constant value $\bar{r}^{\star}=-1$.

To better investigate the benefits in learning the deformed Laplacian form, we considered a dynamic graph whose topology evolves over a window of $40$ time instants $t_k$, $k=1,\ldots,40$.  Specifically, the graph topology is a random graph  at time $t_1=1$ and the graph evolves until it becomes a graph with $3$ clusters at time $t_{20}$ and  a bipartite graph at time $t_{40}$. 
 We then generated $M$ random graph signals
$\bx(i)$, $i=1,\ldots,M$, drawn from a Gaussian distribution, in order to avoid favoring any specific graph spectral representation.\\
Hence, we solved problem $\mathcal{P}$ to explore the sparsity/data-fitting trade-off behavior,  by setting $K=3$ and $\gamma=1$. Specifically,   for each observed signal the following optimization problem is considered
\beq \nonumber
\begin{array}{lll}
 \underset{\bs(i) \in \mathbb{R}^{N \times M}, r \in \mathbb{R} }{\min}  &  \parallel \bx(i) -\mU(r) \bs(i) \parallel_{F}^{2} \medskip\\ 
 \qquad\text{s.t.} 
     &  \parallel \bs(i) \parallel_{0}= K.
\end{array}
\eeq
To solve this problem, Algorithm $1$ is employed   to obtain, via a line search, the optimal parameter $r^{\star}$ and the associated sparse graph signals.

\begin{figure}[t] 
\centering
	\includegraphics[width=8.5cm,height=5.0cm]{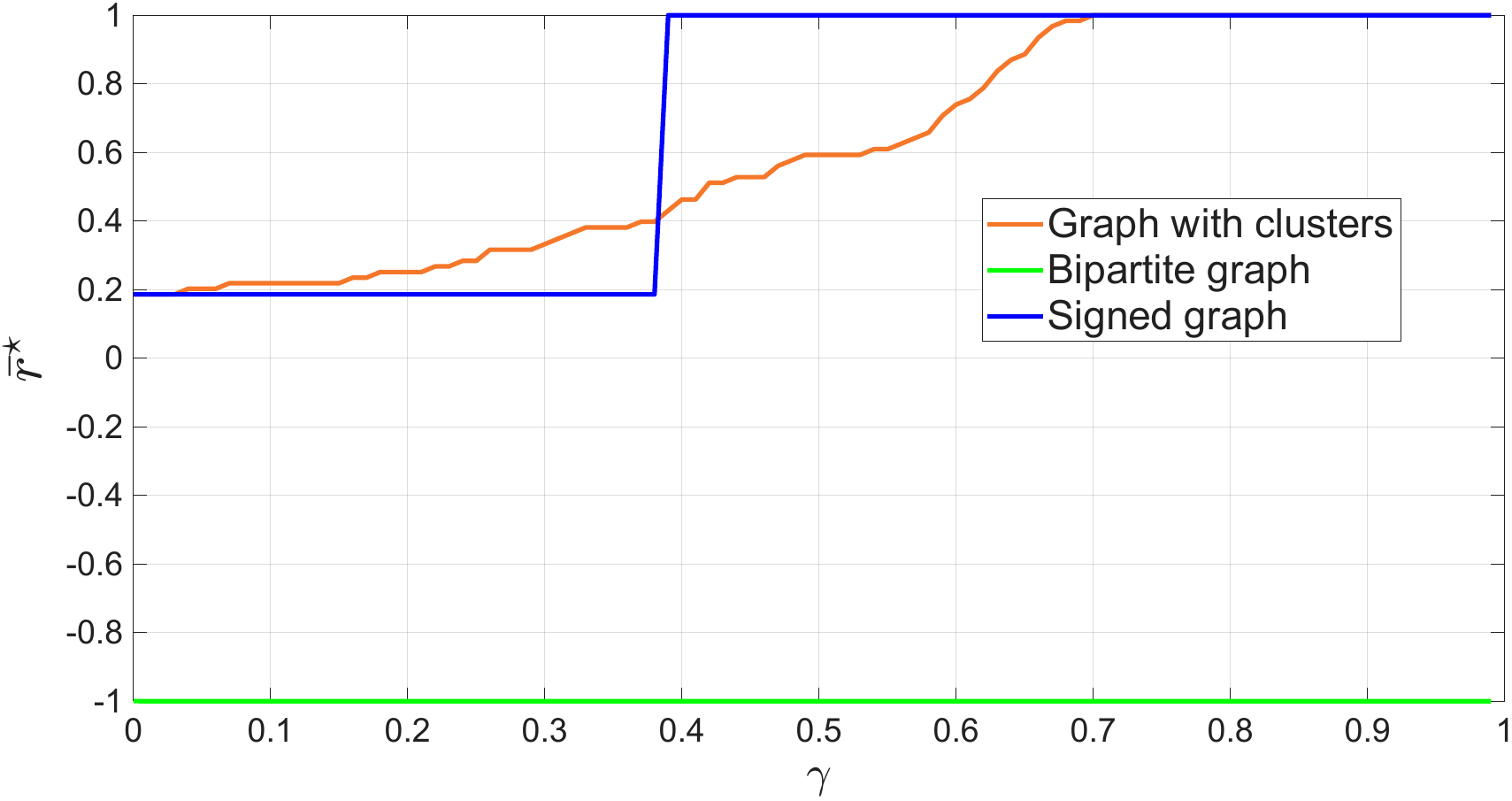}
	  
	\caption{Optimal parameter $r$ versus $\gamma.$}\label{fig:fig_2}
\end{figure}
\begin{figure}[t] 
\centering
	\includegraphics[width=8.5cm,height=5.0cm]{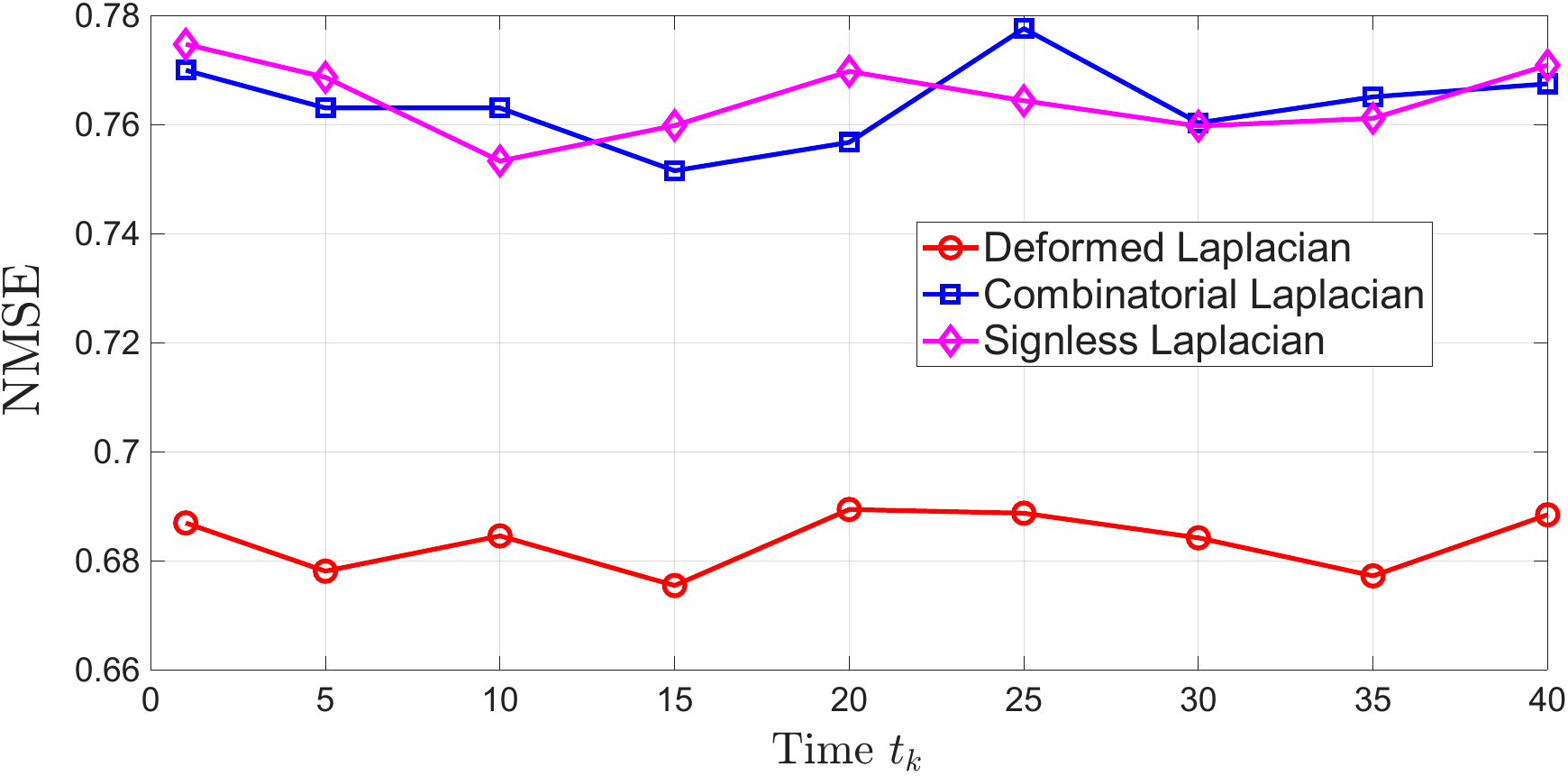}
	    \caption{NMSE versus time for a dynamic graph topology.}\label{fig:fig_3}
\end{figure}

In Fig. \ref{fig:fig_3}  we report the averaged $\text{NMSE}=\sum_{i=1}^M\frac{\parallel \bx(i)-{\hat{\bx}}(i)\parallel_2}{M \parallel \bx(i)\parallel_2 }$, where ${\hat{\bx}}(i)=\mU_{\mathcal{K}}(r^{\star})\bs^{\star}_{\mathcal{K}}$ is the estimated sparse signal. We compare the results obtained by solving $\mathcal{P}$ using as Laplacian forms the combinatorial one and the signless Laplacian. The results highlight that the deformed Laplacian provides a more effective spectral representation, achieving lower reconstruction errors under the same sparsity constraints.

\subsection{Real Data Experiments}
In this section, we apply the developed framework to real-world datasets, showing how the proposed framework yields significant performance gains in both signal representation and recovery. Specifically, we investigate its application to financial data from U.S. stock markets and to the Copenhagen mobile phone communication network.

\textbf{Returns of S$\boldsymbol{\&}$P 500 Stocks.}
As a first real-data application, we considered the  publicly available dataset of  historical daily prices from Yahoo Finance™ for 333 stocks listed in the S$\&$P 500 Index, along with 8 sector indices \cite{de2022learning}. These indexes are  Consumer Staples, Energy, Financials, Health Care, Industrials, Materials, Real Estate, and Utilities. 
Sector labels for each stock are obtained from the Global Industry Classification Standard (GICS) \cite{GICS}. However, predefined sector classifications may be overly restrictive, as stocks can exhibit interactions across multiple sectors. To address this, in \cite{de2022learning} the authors proposed a data-driven bipartite graph learning approach that allows for flexible associations between stocks and sectors.\\
The node set is  composed of  $8$ nodes 
 corresponding to the sector indices,  and $333$ nodes associated with the stocks, for a total number of $N=341$ nodes. We considered the data set and the code provided in \cite{de2022learning} to derive the adjacency matrix $\mA$ using the kSBG method introduced in \cite{de2022learning}. This method learns links between stocks and sectors using  $k=8$ component bipartite estimators. The observed graph signals are 
 the prices $P_{i,j}$ of each stock (node)  
$i$ at day $j$, for $j=1,\ldots,Q$,
 where $Q=1292$ denotes the number of observed daily.
 The observation window is restricted to the period from Jan 1st, 2018 to Jan 1st, 2021, observing $M=84$ days (about $4$ months in terms of stock market days).
  Hence, we derive the 
 log-return  signal matrix $\mX \in \mathbb{R}^{N \times M}$, with entries   computed as $X_{i,j}=\log P_{i,j}-\log P_{i,j-1}$.
\begin{figure}[t] 
\centering
	\includegraphics[width=8.5cm,height=5.0cm]{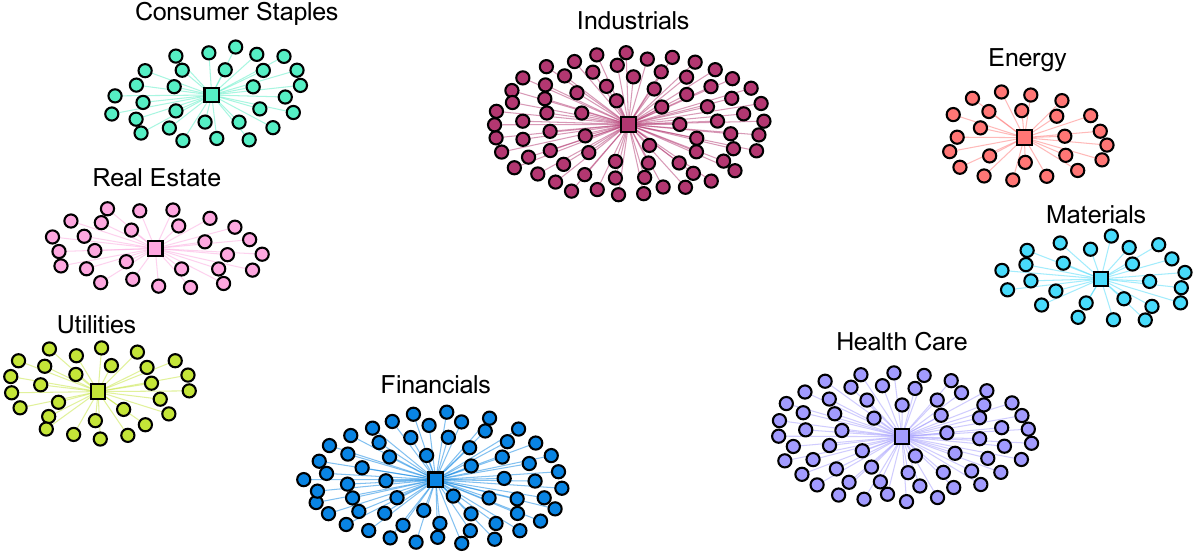}
	    \caption{Financial graph of the U.S. stock markets.}\label{fig:fig_finan_graph}
\end{figure}

 In Fig. \ref{fig:fig_finan_graph} we illustrate the $k=8$ clustered bipartite graph recovered applying the kSBG method proposed in \cite{de2022learning}. Then, we run Algorithm $1$ using the matrices $\mX$ and $\mA$ as inputs.  In Fig. \ref{fig:fig_mse_fin} we illustrate the normalized mean squared error in the recovering of the signals, defined as $\text{NMSE}=\frac{\parallel \mX-\mU_{\mathcal{K}} \mS_{\mathcal{K}}\parallel_F}{\parallel \mX \parallel_F}$ for  $\gamma=0.4$.  It can be  noted a decline in the estimation error starting from Jan 2020 which can be explained considering the impact of the COVID-19 pandemic on the U.S. stock market. Furthermore, the results indicate that the learned Laplacian forms across all  observed temporal windows   do not correspond to any of the standard Laplacian forms, i.e., combinatorial, signless, or signed. 
 \begin{figure}[t] 
\centering
	\includegraphics[width=8.5cm,height=5.2cm]{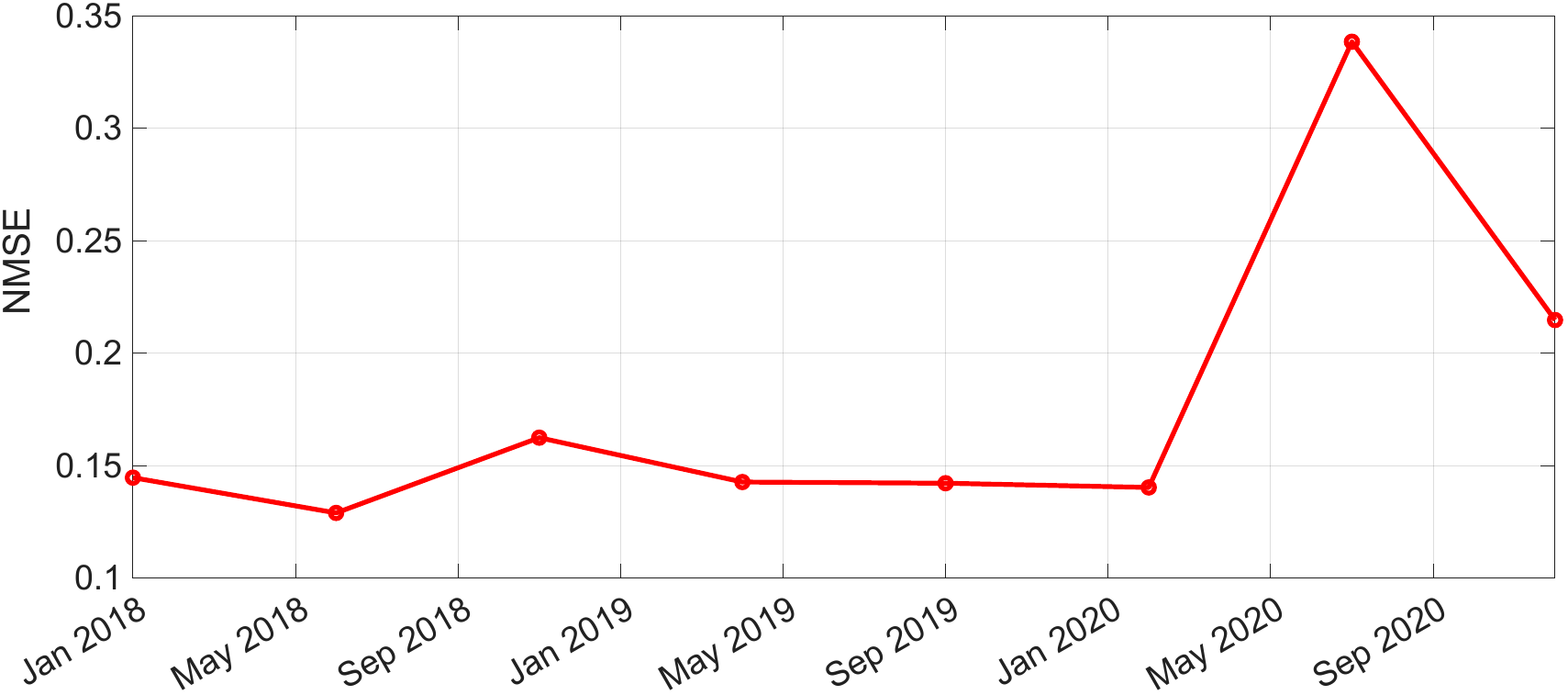}
	    \caption{NMSE versus time for the U.S. stock markets.}\label{fig:fig_mse_fin}
\end{figure}
 
 \textbf{The Copenhagen Communication Networks.}
 As further application to  real-world networks, we consider  the  time-varying network which connects a population of more than 700 university students over a period of four weeks. The dataset was collected via smartphones as part of the Copenhagen Networks Study \cite{copenhagen}.
We focus on the phone call network, where nodes represent users and an edge is established between two users if at least one call occurred between them within a given week. As node signals, we use the total number of outgoing calls from each node to all other nodes, aggregated over a one-week observation window.

Given this graph signal,  we apply our method to recover both the deformed graph Laplacian form and the sparse signal representations.  Fig. \ref{fig:fig_cop2} reports  the NMSE of the signal reconstruction versus the signal sparsity, i.e., for different values of the signal bandwidth $K$ and  varying the coefficient $\gamma$.
The results highlights a clear trade-off between  sparsity and accuracy, since as the number of  spectral signal coefficients increases, the reconstruction error decreases accordingly. Furthermore, as $\gamma$ approaches $1$, the NMSE decreases, as more weight is assigned to the data-fitting error  than to  signal smoothness.
Finally, we observed that for $0 \geq \gamma \leq 1$ the optimal parameters $r$ of the learned deformed Laplacian  does not recover the combinatorial or the signless Laplacian.
\begin{figure}[t] 
\centering
	\includegraphics[width=8.5cm,height=5.1cm]{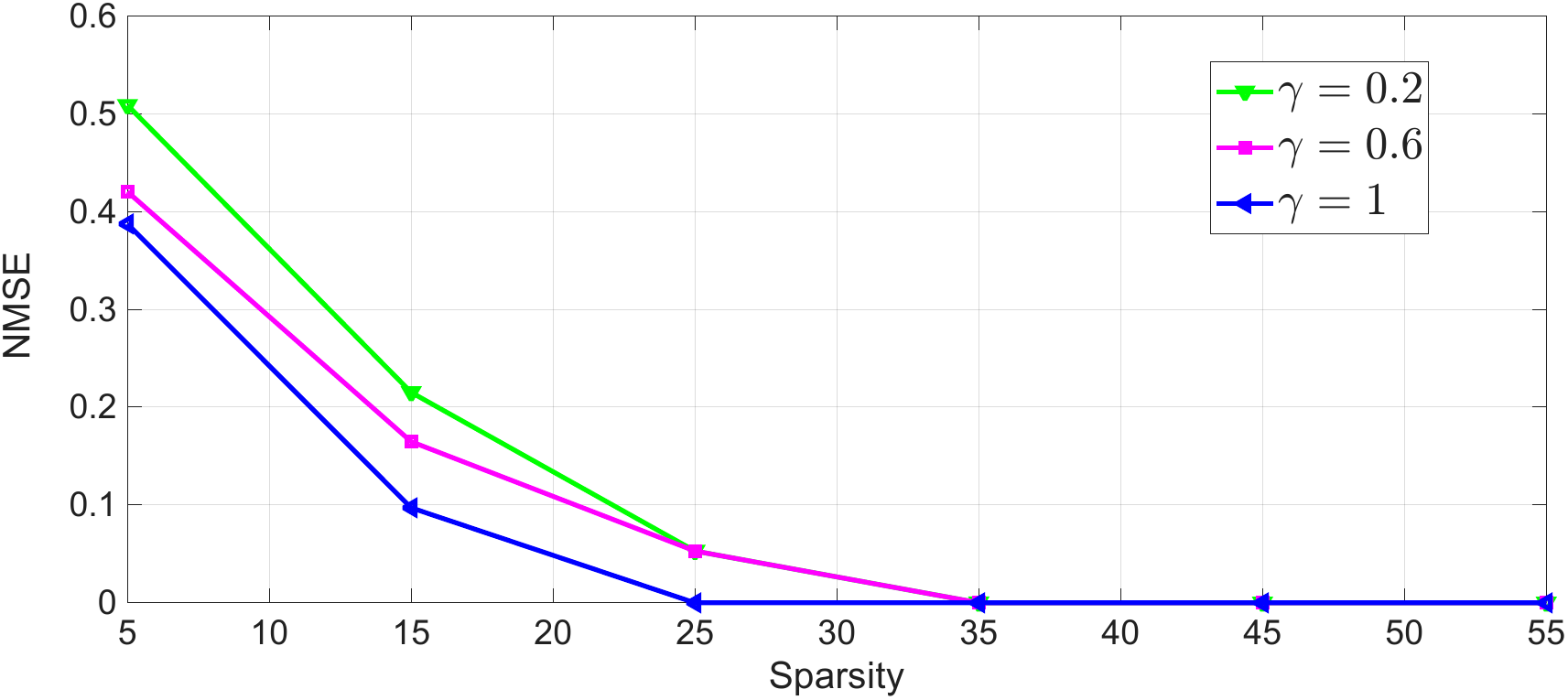}
	    \caption{NMSE versus sparsity for the Copenhagen communication phone calls network.}\label{fig:fig_cop2}
\end{figure}

\section{Conclusions}
In this paper, we extend graph signal processing to graphs represented via the deformed Laplacian, a parametric algebraic operator defined as a second-order matrix polynomial. This formulation provides a unified framework that encompasses different Laplacian forms and enables adaptive modeling of graph structures through a single parameter. 
We have shown that this parametrization allows learning Laplacian operators tailored to both the observed data and the underlying graph topology. In particular, we developed an optimization framework that jointly estimates the graph signal representation and the optimal deformed Laplacian form, thus removing the need for prior knowledge of its form. 
From a dynamical perspective, the deformed Laplacian can be interpreted as the generator of a parametric reaction–diffusion process on graphs, highlighting the role of the deformation parameter in controlling diffusion and local node dynamics.

Extensive numerical experiments on both synthetic and real-world datasets demonstrate improved reconstruction accuracy and enhanced sparsity, particularly in graphs with heterogeneous structures, such as those combining clustered and bipartite patterns. Overall, the proposed framework provides a principled and adaptive approach for learning graph representations, paving the way for data-driven modeling of complex graph-structured data.





\bibliographystyle{IEEEbib}
\bibliography{refs.bib}

\end{document}